\documentclass[11pt]{amsart}

\usepackage{amsmath,amssymb, amsfonts}
\usepackage{amsthm}
\usepackage[utf8]{inputenc}
\usepackage{subfigure}
\usepackage{float}
\usepackage{caption}
\usepackage{graphicx,subfigure}
\newcommand{\sm}{\setminus}
\newcommand{\G}{\mathcal{G}}

\theoremstyle{plain}
\newtheorem{prop}[equation]{Proposition}

\newtheorem{thm}[equation]{Theorem}
\newtheorem{cor}[equation]{Corollary}
\newtheorem{lem}[equation]{Lemma}
\theoremstyle{definition}
\newtheorem{prob}[equation]{Problem}
\newtheorem{defn}[equation]{Definition}
\newtheorem{rem}[equation]{Remark}

\DeclareMathOperator*{\EFC}{EFC}
\DeclareMathOperator*{\ESE}{ESE}

\DeclareMathOperator*{\ECE}{ECE}

\DeclareMathOperator*{\Shed}{Shed}

\usepackage{tikz}
\usepackage{verbatim}
\usepackage{tikz}
\usetikzlibrary{shapes,shadows,calc}
\usepgflibrary{arrows}
\usetikzlibrary{arrows, decorations.markings, calc, fadings, decorations.pathreplacing, patterns, decorations.pathmorphing, positioning}
\tikzset{nodde/.style={circle,draw=blue!50,fill=pink!80,inner sep=4.2pt}}
\tikzset{noddee/.style={circle,draw=black,fill=black,inner sep=2pt}}
\tikzset{noddee2/.style={circle,draw=black,fill=black,inner sep=3pt}}
\usepackage{scalefnt}
\usetikzlibrary{arrows,decorations.pathmorphing,backgrounds,positioning,fit,petri}

\newtheorem{claim}{Claim}
\newenvironment{cproof}[1]{\par\indent{\textit{Proof of the Claim.}}\space#1}{\hfill $\diamondsuit$}

\newtheorem{case}{Case}
\usepackage{changes}
\definechangesauthor[name=Zakir, color=blue!50!black]{zakir}

\begin{document}
\title{Edge-Stable Equimatchable Graphs}
\author{Zakir Deniz}\thanks{Department of Mathematics, Duzce University, Duzce, Turkey}
\author{T{\i}naz Ekim}\thanks{Department of Industrial Engineering, Bo\u{g}azi\c{c}i University, Istanbul, Turkey }
\date{\today}

\begin{abstract} 
 A graph $G$ is \emph{equimatchable} if every maximal matching of $G$ has the same cardinality. We are interested in equimatchable graphs such that the removal of any edge from the graph preserves the equimatchability. We call an equimatchable graph $G$ \emph{edge-stable}  if $G\setminus {e}$, that is the graph obtained by the removal of edge $e$ from $G$, is also equimatchable for any $e \in E(G)$. After noticing that edge-stable equimatchable graphs are either 2-connected factor-critical or bipartite, we characterize edge-stable equimatchable graphs. This characterization yields an $O(\min(n^{3.376}, n^{1.5}m))$ time recognition algorithm. Lastly, we introduce and shortly discuss the related notions of edge-critical, vertex-stable and vertex-critical equimatchable graphs. In particular, we emphasize  the links between our work and the well-studied notion of shedding vertices, and point out some open questions.  
\end{abstract}

\maketitle

\section*{Introduction}
Given a graph $G$, a \emph{matching} is a set of edges of $G$ having pairwise no common endvertex. A matching $M$ is called \emph{maximal} if no new edge can be added to $M$ while keeping its property of being a matching. The problem of finding an inclusion-wise maximal matching of minimum size in a graph has been a central problem for many researchers both for its practical and theoretical point of views. This minimum size is equal to the size of a minimum edge dominating set which is also widely studied \cite{yannakakis}. One application is the following: let $A$ be a $0-1$ matrix, and consider the problem of finding a minimum set $C$ of 1's in $A$ such that any other 1 of $A$ is in the same row or column with an element of $C$. Another application is about a telephone switching network built to route phone calls from incoming lines to outgoing trunks (assuming that a trunk can
pass only one phone call at a time). The problem is to find the worst-case behavior of the
network, i.e., the minimum number of routed calls when the network is saturated, thus
no new call can be routed. Both applications can be modeled as the problem of finding a minimum maximal matching in a bipartite graph. However, solving this problem is NP-hard even in bipartite graphs with maximum degree 3 \cite{yannakakis}, or in $k$-regular bipartite graphs for any fixed $k\geq 3$ \cite{DET}. We note that finding a minimum maximal matching becomes a trivial task if all the maximal matchings of the graph under consideration have the same size. In this case, any maximal matching constructed greedily is a minimum (and maximum) one.  

A graph $G$ is called \emph{equimatchable} if every maximal matching of $G$ has the same cardinality. Equimatchable graphs has attracted a lot of attention in literature; they are mainly considered from structural point of view (see for instance \cite{ClawfreeEqm, eqm_regular, oddcycles-eqm, EK2, Favaron, k-eqm, genus, Plummer}). In this paper, we study equimatchable graphs from another structural perspective; we deal with the stability of this desired property of being equimatchable with respect to edge removals. 
Formally, an equimatchable graph $G$ is called \emph{edge-stable} if $G\setminus {e}$ is equimatchable for each $e \in E(G)$. Edge-stable equimatchable graphs are denoted \emph{$\ESE$-graphs} as a shorthand. 
Conventionally, we assume that a graph which consists of a single vertex is both equimatchable and $\ESE$. Consequently, a graph consisting of a single edge is also $\ESE$.


We note that equimatchable graphs are closely related to a well-studied graph class called well-covered graphs. Given a graph $G$, a set of vertices $S\subseteq V(G)$ is an \emph{independent set} if vertices of $S$ are pairwise nonadjacent. An independent set is said to be  \emph{maximal} if no other independent set properly contains it. A graph is called \emph{well-covered} if all its maximal independent sets have the same size. It is an easy observation to see that a graph $G$ is equimatchable if and only if its line graph $L(G)$ is well-covered where $L(G)$ is obtained by replacing every edge of $G$ with a vertex in $L(G)$ and where two vertices of $L(G)$ are adjacent if and only if their corresponding edges in $G$ share a common end-vertex. In other words, recognizing equimatchable graphs is equivalent to recognizing well-covered line graphs. It is worth mentioning that although the recognition of well-covered graphs is co-NP-complete \cite{chvatal, SS}, well-covered line graphs (and therefore line graphs of equimatchable graphs) can be recognized in polynomial time \cite{demange-ekim-equi}. In 1979, Staples introduced the class of $W_2$ graphs, better known as 1-well-covered graphs (without isolated vertices); this class coincides with well-covered graphs that remain well-covered upon removal of any vertex \cite{StaplesThesis, Staples79}. We note that, a graph is ESE if and only if its line graph is 1-well-covered graph. The excellent survey of Plummer on well-covered graphs dated 1992 already contains several results on 1-well-covered graphs  \cite{survey}. After the study of some basic properties of 1-well-covered graphs in \cite{StaplesThesis, Staples79}, several papers, including some recent ones, studied this class. We note that like well-covered graphs, the recognition of 1-well-covered graphs is also in co-NP \cite{survey}, however, the complexity of their recognition is unknown to the best of our knowledge. Later on, several papers focused on subclasses of 1-well-covered graphs. In \cite{pinter92} and \cite{pinter95}, Pinter gives the characterization of respectively  4-regular planar 3-connected 1-well-covered graphs, and 1-well-covered graphs which are planar and of girth 4. In \cite{pinter97}, he also provides constructions of infinite families of 1-well-covered graphs of girth 4. Later in \cite{hartnell2006}, Hartnell gives the characterization of 1-well-covered graphs with no 4-cycles. More recently, a characterization of 1-well-covered graphs where every triangle in $G$ is also a dominating set for $G$ is given in \cite{HT2016}. Lastly, Levit and Mandrescu give characterizations of all 1-well-covered graphs in terms of the existence of special independent sets \cite{LM2016}. However, these characterizations rely on the existence of independent sets with some properties that can not be checked in polynomial time. In addition, to the best of our knowledge, none of the above mentioned papers discuss the possibility of using the characterization obtained for a subclass of 1-well-covered graphs for the development of a recognition algorithm. In this paper, we characterize $\ESE$-graphs whose recognition is equivalent to the recognition of 1-well-covered line graphs. In addition, we show that our characterization yields an efficient recognition algorithm. We also briefly mention some related notions such as edge-criticality, vertex-stability and vertex-criticality of equimatchable graphs.

We start in Section \ref{sec:prem} with some definitions and preliminary results on $\ESE$-graphs. Refer to these definitions for various terms used in the following description of our contribution. As justified by Theorem \ref{Thm: all 2 connected graphs}, we divide $\ESE$-graphs into three categories: 2-connected factor-critical $\ESE$-graphs in Section \ref{sec:FC ESE}, $\ESE$-graphs with a cut vertex in Section \ref{sec: ESE-graphs with a cut vertex}, and bipartite $\ESE$-graphs in Section \ref{sec: bip ESE}. These results provide a full characterization of all $\ESE$-graphs yielding an $O(\min(n^{3.376}, n^{1.5}m))$ time recognition algorithm (in Section \ref{sec: Recog ESE}) which is better than the most natural way of recognizing $\ESE$-graphs by checking the equimatchability of $G\setminus e$ for every $e \in E$. Lastly, in Section \ref{sec: Conc remarks}, we first interpret our results in terms of some recent works on 1-well-covered graphs and the so-called shedding vertices. Motivated by this, we introduce the opposite notion of edge-critical equimatchable graphs (called $\ECE$-graphs) which are minimally equimatchable graphs with respect to edge removals. We conclude by giving some insight on our ongoing work about $\ECE$-graphs and pointing out some research directions.


\section{Definitions and Preliminaries}\label{sec:prem}

Given a graph $G=(V,E)$ and a subset of vertices $I$, $G[I]$ denotes the subgraph of $G$ induced by $I$, and $G\setminus I=G[V\setminus I]$. When $I$ is a singleton $\{v\}$, we denote $G\setminus I$ by $G - v$. We also denote by $G\setminus e$ the graph $G(V,E\setminus\{e\})$. For a subset $I$ of vertices, we say that $I$ is \emph{complete} to another subset $I'$ of vertices (or by abuse of notation, to a subgraph $H=G[I]$) if all vertices of $I$ are adjacent to all vertices of $I'$ (respectively $H$). Note that this is  symmetric: if $I$ is complete to $I'$ then $I'$ is also complete to $I$. We denote by $K_n$ and $C_n$ respectively, a complete graph and a cycle on $n$ vertices. For a vertex $v$, the neighborhood of $v$ in a subgraph $H$ is denoted by $N_H(v)$. We omit the subscript $H$ whenever it is clear from the context. For a subset $V'\subseteq V$, we have $N(V')=(\cup_{v\in V'}N(v))\setminus V'$. We also use the notation $[k]$ to denote the set of integers $\{1,2,\ldots ,k\}$.

Given a graph $G$, the size of a maximum matching of $G$ is called the \emph{matching number} of $G$ and denoted by $\nu(G)$. A matching is \emph{maximal} if no other matching properly contains it. A matching $M$ is said to \emph{saturate} a vertex $v$ if it is the endvertex of an edge in $M$, otherwise $M$ is said to leave vertex $v$ \emph{exposed}. A matching saturating all vertices of $G$ is called a \emph{perfect matching}. If every matching of $G$ extends to a perfect matching, then  $G$ is called \emph{randomly matchable}. Clearly, equimatchable graphs having a perfect matching are exactly randomly matchable graphs. These graphs have been characterized by Sumner \cite{Sumner}.
\begin{lem}\cite{Sumner}\label{lem:randomly}
A connected graph is randomly matchable if and only if it is isomorphic to either $K_{2r}$ or $K_{r,r}$ for some $r\geq 1$.
\end{lem}

We note that $G$ is $\ESE$ if and only if every connected component of $G$ is $\ESE$. Consequently, since $K_2$ is $\ESE$ by convention, a graph consisting of only connected components isomorphic to $K_2$ is $\ESE$.

 The following is a direct translation of a result in \cite{Staples79} for 1-well-covered graphs with no isolated vertices (called the class $W_2$ in the original paper) in terms of $\ESE$-graphs.

\begin{lem}\cite{Staples79}\label{lem:staples79}
A graph $G$  with no connected component isomorphic to $K_2$  is $\ESE$ if and only if $\nu(G \sm e) = \nu(G)$ and $G \sm e$ is equimatchable for every $e \in E(G)$. 
\end{lem}



The following shows that we do not need to require the graph $G \sm e$ to be equimatchable in Lemma \ref{lem:staples79} if we state it only in one direction. Moreover, we can relax the condition that $G$ has no connected component isomorphic to $K_2$ by appropriately choosing the edge $e$ to be removed. We will heavily use Proposition \ref{prop:unchanged matching number} and Corollary \ref{rem:unchanged matching number} in our proofs.

\begin{prop}\label{prop:unchanged matching number}
Let $G$ be an $\ESE$-graph. Then $\nu(G)=\nu(G\setminus e)$ for every $e\in E(G)$ such that the endpoints of $e$ do not form a connected component of $G$.

\end{prop}

\begin{proof}
Let $G$ be an $\ESE$-graph and assume for a contradiction that $\nu(G)=\nu(G\setminus e)+1$ for some $e \in E(G)$ whose endpoints do not form a component of $G$. It follows that there is a maximum, thus, maximal matching $M$ (of $G$) containing $e$ for $e=uv \in E(G)$ such that $M \sm e$ is also a maximal matching in $G \sm e$. Since $G[\{u,v\}]$ does not induce $K_2$, there exists $w \in N(u)\cup N(v)$, without loss of generality say $wu \in E(G)$. Moreover $wu$ can be extended to a maximal matching $M'$ of $G$, thus by equimatchability of $G$,  $|M|=|M'|$. Besides, $M'$ is also a maximal matching in $G \sm e$ of size $\nu(G)$, contradiction. 
\end{proof}

 A consequence of Proposition \ref{prop:unchanged matching number} is the following.

\begin{cor}\label{rem:unchanged matching number}
The only connected $\ESE$-graph with a perfect matching is $K_2$.  
\end{cor}

\begin{proof}
Let $G$ be a connected $\ESE$-graph different from $K_2$. Let  $M$ be a maximal matching of $G$ and $e=uv\in E(G)$. By Proposition \ref{prop:unchanged matching number},
$M\sm e$ is not maximum in $G\sm e$, and thus not maximal neither since $G\sm e$ is equimatchable. Therefore, $M\sm e$ can be extended by adding an edge $e'$, which is necessarily adjacent to $e$. It follows that there exists a vertex $w$ which is the endpoint of $e'$
different from $u, v$ such that $w$ is exposed by $M$. Therefore $M$ is not a perfect matching of $G$.
\end{proof}

Recall that a graph is $\ESE$ if and only if every connected component of it is $\ESE$. Therefore, in the remainder of this paper, we consider only connected $\ESE$-graphs. A connected graph $G$ is said to be \emph{$k$-connected} if it has at least $k+1$ vertices and at least $k$ vertices should be removed to make it disconnected. Given a connected graph $G$, a vertex $v\in V(G)$ is called a \emph{cut vertex} if $G-v$ is disconnected. If $G - v$ has a perfect matching for each $v \in V(G)$, then $G$ is called \emph{factor-critical}. For short, an equimatchable factor-critical  graph is called an \emph{$\EFC$-graph}. 

The following is a consequence of the results in \cite{Plummer} (although it is not explicitly mentioned in this paper) and will guide us through our characterization. 


\begin{thm}\cite{Plummer} \label{Thm: all 2 connected graphs}
A 2-connected equimatchable graph is either factor-critical or bipartite or $K_{2t}$ for some $t\geq 2$.
\end{thm}

We aim at characterizing all $\ESE$-graphs apart from $K_1$ and $K_2$ which are assumed to be $\ESE$ by convention. One can easily observe that for $t\geq 2$, the equimatchability of $K_{2t}$ is lost when an edge is removed, hence, $K_{2t}$ is not $\ESE$. It follows that $\ESE$-graphs can be studied under three categories: 2-connected factor-critical, 2-connected bipartite and those having a cut vertex. Note that a bipartite graph can not be factor-critical, hence, 2-connected factor-critical $\ESE$-graphs and 2-connected bipartite $\ESE$-graphs form a partition of 2-connected $\ESE$-graphs into two disjoint subclasses. On the other hand, an $\ESE$-graph with a cut vertex could be either factor-critical, or bipartite, or none of them. Therefore, a separate characterization of each one of these three categories would lead to a full characterization of all $\ESE$-graphs with possibly some overlaps (some graphs belonging to two categories).  However, by showing in Section \ref{sec: ESE-graphs with a cut vertex} that $\ESE$-graphs with a cut vertex are bipartite, we provide a characterization of all $\ESE$-graphs containing only two exclusive cases: (2-connected) factor-critical $\ESE$-graphs (Section \ref{sec:FC ESE}) and  bipartite $\ESE$-graphs (Section \ref{sec: bip ESE}).




\section{Factor-critical ESE-graphs} \label{sec:FC ESE} 

Let us first underline that although we seek to characterize factor-critical $\ESE$-graphs which are  2-connected, all the results in this section are valid for any factor-critical $\ESE$-graph. Note that factor-critical graphs are connected but not necessarily 2-connected. However, it turns out that factor-critical $\ESE$-graphs are also 2-connected (See Corollary \ref{cor: FC ESE are 2conn}).

Remind that a factor-critical graph $G$ has $\nu(G)=(\vert V(G) \vert-1)/2$ and if it is equimatchable then all maximal matchings have size  $(\vert V(G) \vert-1)/2$. 

\begin{lem} \label{lem: defn equim}
Let $G$ be a factor-critical graph. $G$ is equimatchable if and only if there is no independent set $S$ such that $|S|=3$ and $G\setminus S$ has a perfect matching.
\end{lem}
\begin{proof}
Let $G$ be an $\EFC$-graph. Then each maximal matching is of size $(\vert V(G) \vert-1)/2$. If there is an independent set $S$ with $\vert S \vert = 3$ such that $G\setminus S$ has a perfect matching $M$, then $M$ is also a maximal matching in $G$ and has size strictly less than $(\vert V(G)\vert -1 )/2$, a contradiction with being equimatchable. \\
Now, we suppose the converse. That is, for all independent set $S$ with $\vert S\vert = 3$, $G \setminus S$ has no perfect matching. Assume $G$ is not equimatchable and admits therefore a maximal matching of size strictly less than $\nu(G)=(\vert V(G) \vert-1)/2$.  Remark that, if there is a maximal matching $M$ such that $\vert M \vert \leq \nu(G)-2$, then there is also a maximal matching $M'$ of size $\nu(G)-1$ which can be obtained from $M$ by repetitively using augmenting chains (whose existence are guaranteed by the fact that the matching under consideration is not maximum). Hence, there are exactly 3 vertices exposed by the matching $M'$. They form an independent set $S$ in $G$ and $M'$ is a perfect matching in $G\setminus S$. This is a contradiction. So, $G$ is equimatchable.
\end{proof}



The following equivalence for an $\EFC$-graph to be edge-stable will be very useful. Note that factor-critical graphs have no component isomorphic to $K_2$, thus Proposition \ref{prop:unchanged matching number} applies to every edge of an edge-stable $\EFC$-graph.

\begin{lem} \label{lem: defn edg-hered equim}
Let $G$ be an $\EFC$-graph. Then $G$ is edge-stable if and only if there is no induced $\overline{P_3}$ in $G$ such that $G\setminus \overline{P_3} $ has a perfect matching.
\end{lem}
\begin{proof}
Let $G$ be an $\EFC$-graph which is also edge-stable. Assume $G$ has an induced $\overline{P_3}$ on vertices $\{v,u_1,u_2\}$ with an edge between $u_1$ and $u_2$ such that $G\setminus \{v,u_1,u_2\} $ has a perfect matching $M$. Then, $M$ is a maximal matching of $G \sm u_1u_2$ of size one less than the matching $M\cup \{u_1u_2\}$ of $G$, contradicting the edge-stability of $G$ by Proposition \ref{prop:unchanged matching number}.
\\
Now, let us consider the converse. Assume for a contradiction that $G$ is not edge-stable. Then there is at least one edge $u_1u_2$ such that $G \setminus u_1u_2$ is not equimatchable. This means, in particular, that there is a maximal matching of $G \setminus u_1u_2$ leaving $u_1,u_2$ and one more vertex, say $v$, exposed (remind that all maximal matchings of $G$ leave exactly one vertex exposed). Therefore, $v$ is not adjacent to $u_1$ and $u_2$, thus $G[\{v,u_1,u_2\}]\cong\overline{P_3}$ and $M\setminus \{u_1u_2\}$ is a perfect matching of $G\setminus \{v,u_1,u_2\}$, a contradiction. 
\end{proof}

Although it is not directly related to further results, the following gives some insight about the structure of $\ESE$-graphs which are factor-critical. Remind that $diam(G)= \max \{d(u,v)| u,v\in V(G)\}$ where $d(u,v)$ is the \emph{distance} (i.e., the length of the shortest path) between vertices $u$ and $v$.

\begin{cor} \label{cor: edge hered diam 2}
Let $G$ be an $\EFC$-graph. If $G$ is edge-stable, then $diam(G)\leq 2$.
\end{cor}
\begin{proof}
Consider an $\EFC$-graph $G$ which is also edge-stable. Then, for each vertex $v\in V(G)$, there is a matching $M_v$ leaving $v$ as the only exposed vertex. By Lemma \ref{lem: defn edg-hered equim}, $v$ is adjacent to at least one of the endpoints of each edge in $M_v$, since otherwise $G$ contains  an induced subgraph $\overline{P_3}$ such that $G \setminus \overline{P_3}$ has a perfect matching. Since $M_v$ saturates all vertices except $v$, we have  $d(v,u)\leq 2$ for each $u\in V(G)$. By selecting $v$ arbitrarily, we have  $d(v,u)\leq 2$ for every pair of vertices $u,v\in  V(G)$ and therefore, $diam(G)\leq 2$.
\end{proof}

In what follows, we will be using a special decomposition of a factor-critical $\ESE$-graph. Let $G$ be a factor-critical $\ESE$-graph, and $M_v$ a perfect matching of $G-v$ for some $v\in V(G)$. We note that there is no edge $yz \in M_v$ such that  $\{y,z\}\cap N(v)= \emptyset$ since otherwise $\{v,y,z\}$ induces a $\overline P_3$ whose removal leaves the perfect matching $M_v\sm \{yz\}$ of $G\sm \{v,y,z\}$, contradicting the edge-stability of $G$ by Lemma \ref{lem: defn edg-hered equim}. Therefore, the following canonical decomposition exists for any factor-critical $\ESE$-graph.

\begin{defn}
Let $G$ be a factor-critical $\ESE$-graph. Then, a \emph{canonical decomposition $(N_1,N'_1,N_2)$ of $G$} with respect to a vertex $v$ and a perfect matching $M_v$ of $G-v$ is defined as follows (see Figure \ref{fig:minimal matching M_v isolating v}):
\begin{itemize}
\item $N_1=\{u_1, \ldots , u_t\}$ is the set of neighbors of $v$ which are matched with non-neighbors of $v$, 
\item $N_1'=\{u'_1, \ldots , u'_t\}$ is the set of vertices matched to $N_1$, and
\item $N_2=\{x_1,x'_1, \ldots , x_p, x'_p\}$ is the set of neighbors of $v$ matched to other neighbors of $v$ with matching edges $\{x_1x'_1, \ldots , x_px'_p\}$,
\end{itemize}
where  $V(M_v)=N_1\cup N_1' \cup N_2$ holds.
\end{defn}

Let us emphasize that although the sets $N_1, N'_1, N_2$ are defined with respect to $v$ and $M_v$, we do not adopt an additional index to denote this for the sake of simplicity.

The following lemma is an important intermediary result towards the achievement of our goal. We call an independent set $S$ \emph{nontrivial} if $|S|\geq 2$.

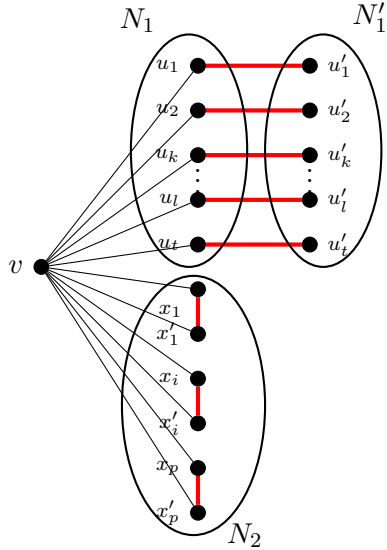
\begin{figure}[hbt] 
\begin{center}
 \resizebox{0.35\textwidth}{!}{%
\begin{tikzpicture}[scale=.6]

\node [noddee] at (2,4.5) (u1) [label=left:{\scriptsize $u_1$}] {};
\node [noddee] at (2,3.5) (u2) [label=left:{\scriptsize $u_2$}] {};
\node [noddee] at (2,2.5) (ui)[label=left:{\scriptsize $u_k$}] {};
\node [noddee] at (2,1.5) (uj) [label=left:{\scriptsize $u_l$}] {};
\node [noddee] at (2,.5) (ut) [label=left:{\scriptsize $u_t$}] {};
\node [noddee] at (2,-.5) (x1) [label=below left:{\scriptsize $x_1$}] {};
\node [noddee] at (2,-1.5) (x11) [label=left:{\scriptsize $x_1'$}] {}
	edge [ultra thick, red] (x1);
\node [noddee] at (2,-2.5) (xj) [label=left:{\scriptsize $x_i$}] {};
\node [noddee] at (2,-3.5) (xjj) [label=left:{\scriptsize $x_i'$}] {}
	edge [ultra thick, red] (xj);
\node [noddee] at (2,-4.5) (xp) [label=left:{\scriptsize $x_p$}] {};
\node [noddee] at (2,-5.5) (xpp) [label=left:{\scriptsize $x_p'$}] {}
	edge [ultra thick, red] (xp);
\node [noddee] at (-1.5,0) (v) [label=left:$v$] {}
	edge [] (u1)
	edge [] (u2)
	edge [] (ui)
	edge [] (uj)
	edge [] (ut)
	edge [] (x1)
	edge [] (x11)
	edge [] (xj)
	edge [] (xjj)
	edge [] (xp)
	edge [] (xpp);	
\node [noddee] at (4.5,4.5) (u11) [label=right:{\scriptsize $u_1'$}] {}
	edge [ultra thick, red] (u1);
\node [noddee] at (4.5,3.5) (u22) [label=right:{\scriptsize $u_2'$}] {}
	edge [ultra thick, red] (u2);
\node [noddee] at (4.5,2.5) (uii)[label=right:{\scriptsize $u_k'$}] {}
	edge [ultra thick, red] (ui);
\node [noddee] at (4.5,1.5) (ujj) [label=right:{\scriptsize $u_l'$}] {}
	edge [ultra thick, red] (uj);
\node [noddee] at (4.5,.5) (utt) [label=right:{\scriptsize $u_t'$}] {}
	edge [ultra thick, red] (ut);
\draw[thick] (1.8,2.6) ellipse (1.3cm and 2.6cm);	
\draw[thick](4.8,2.6) ellipse (1.3cm and 2.6cm);	
\draw[thick] (1.9,-3.1) ellipse (1.6cm and 2.9cm);		
\node at (1.5,4.8) (N1) [label=above left:$N_1$] {};		
\node at (5,4.8) (N11) [label=above right:$N_1'$] {};			
\node at (2.2,-6) (N2) [label=right:$N_2$] {};	
\node at (2,2.1) (dot) {$\vdots$};
\node at (4.5,2.1) (dot) {$\vdots$};		
\end{tikzpicture}
}

\end{center}
\caption{A canonical decomposition of a factor-critical $\ESE$-graph $G$ which is not an odd clique, for some vertex $v\in V(G)$ such that $d(v)<|V(G)-1|$ and a perfect matching $M_v$ of $G-v$ which is shown with bold (red) edges.}
\label{fig:minimal matching M_v isolating v}
\end{figure}

\setcounter{claim}{0}
\begin{lem} \label{lem: exist a clique S}
Let $G$ be a factor-critical graph with at least $7$ vertices. If $G$ is an $\ESE$-graph which is not an odd clique, then there is a nontrivial independent set $S$ which is complete to $G\sm S$.
\end{lem}
\begin{proof}
Assume $G$ is a factor-critical $\ESE$-graph with at least $7$ vertices and which is not an odd clique, then let us show that $G$ contains an independent set $S$ with $\vert S \vert \geq 2$ which is complete to $G\sm S$.  Consider a canonical decomposition $(N_1,N'_1,N_2)$ of $G$ with respect to some vertex $v$ such that $d(v)<|V(G)-1|$ and a perfect matching $M_v$ of $G-v$. As $d(v)<|V(G)-1|$, clearly $N_1\neq \emptyset$ and $N_1'\neq \emptyset$.

The first step of the proof is to show that $N_1$ is an independent set and $N_1'$ is complete to $N_1$ (Claim 2). After that, we split the proof into two cases: $N_1'$ is independent or not. We show that if $N_1'$ is not an independent set, then $S=N_1=\{u_1,u_2,\ldots,u_t\}$ is a nontrivial independent set which is complete to $G\sm N_1$.  The rest of the proof is dedicated to the case where $N_1'$ is an independent set. In this case, we show that
\begin{itemize}
\item if  $N_2$ is complete to $N_1'$, then $S=N_1' \cup \{v\}$ is an independent set complete to $G \sm S$,
\item else  $S=\{u_1,u_2,\ldots,u_t\}\cup\{x_1',x_2',\ldots,x_p'\}$ is an independent set complete to $G\sm S$.
\end{itemize}



\noindent
We now show that $N_1$ is an independent set. For otherwise, there are $ u_k$ and $u_l$ in $N_1$ such that $u_ku_l \in E(G)$. If $u_k'u_l' \notin E(G)$, then $\{u_k',u_l',v\}$ is an independent set $I$ of size $3$ and  $ (M_v\setminus \{u_ku_k',u_lu_l'\} ) \cup \{u_ku_l\} $ is a perfect matching of $G \setminus I$, implying that $G$ is not equimatchable by Lemma \ref{lem: defn equim}. If $u_k'u_l' \in E(G)$, then $\{u_k',u_l',v\}$ induce a $\overline{P_3}$ in $G$ such that $G \setminus \overline{P_3}$ has the same perfect matching as previously, contradicting that $G$ is edge-stable by Lemma \ref{lem: defn edg-hered equim}. It follows that $N_1$  is an independent set.\medskip

\begin{claim}\label{cl:N2 empty}
If $N_1'$ is not an independent set then $N_2= \emptyset$.
\end{claim}
\begin{cproof}
Assume for a contradiction that $N_1'$ is not independent but $N_2 \neq \emptyset$. We first claim that for an edge $u_k'u_l'$ in $N_1'$, each one of $u_k$ and $u_l$ is adjacent to at least one of $x_i$ and $x'_i$ for an edge $x_ix_i' \in M_v$. Indeed, $u_k$ is adjacent to $x_i$ or $x_i'$, since otherwise $G[\{u_k,x_i,x_i'\}]=\overline{P_3}$ and $ M_v\setminus \{x_ix_i',u_ku_k',u_lu_l'\} ) \cup \{vu_l,u_k'u_l'\} $ gives a perfect matching in $G\sm \{u_k,x_i,x_i'\}$, a contradiction by Lemma \ref{lem: defn edg-hered equim}. Similarly, $u_l$ is adjacent to $x_i$ or $x_i'$. 

Now, there are two possible cases. 
Suppose first $x\notin N(u_k)\cup N(u_l)$ for some $x\in \{x_i,x_i'\}$. Then, without loss of generality $\{u_k,u_l\} \in N(x_i)$. In this case, $\{u_k,u_l,x_i'\}$ is an independent set $I$ and  $( M_v\setminus \{x_ix_i',u_ku_k',u_lu_l'\} ) \cup \{vx_i,u_k'u_l'\} $ gives a perfect matching in $G\setminus I$, a contradiction with $G$ being equimatchable by Lemma \ref{lem: defn equim}. 
Otherwise, we have without loss of generality $x_iu_k,x_i'u_l \in E(G)$, then $G[\{v,u_k',u_l'\}]=\overline{P_3}$ and  $( M_v\setminus \{x_ix_i',u_ku_k',u_lu_l'\} ) \cup \{x_iu_k,x_i'u_l\} $ gives a perfect matching of $G \sm \{v,u_k',u_l'\}$, a contradiction by Lemma \ref{lem: defn edg-hered equim}. As both cases are concluded with a contradiction, it follows that $N_2 = \emptyset$. 
\end{cproof}\\

\begin{claim}\label{cl:complete}
$N_1'$ is complete to $N_1$.
\end{claim}
\begin{cproof}
Recall that $t=|N_1|$. The case of $t=1$ holds trivially. So assume that $t\geq2$ and the claim is false, that is,  there exist $u_k \in N_1$, $u_l' \in N_1'$  such that $u_ku_l' \notin E(G)$.
If $u_k'u_l'\notin E(G)$,  then $G[\{u_k,u_k',u_l'\}]=\overline{P_3}$ and $G \setminus \{u_k,u_k',u_l'\}$ has a perfect matching $(M_v\setminus \{u_ku_k',u_lu_l'\} ) \cup \{vu_l\} $, a contradiction by Lemma \ref{lem: defn edg-hered equim}. So $u_k'u_l'\in E(G)$ and therefore, $N_1'$  is not an independent set. We remark that in this case $t \geq 3$ since $\vert V(G) \vert=2r+1$ for $r\geq 3$ and $N_2=\emptyset$ by Claim \ref{cl:N2 empty}. The followings hold by Lemma \ref{lem: defn edg-hered equim}.
\begin{itemize}
\item $u_ku_s'\in E(G)$ for all $s\in[t]\setminus\{k,l\}$, since otherwise $G[\{u_k,u_s,u_s'\}]=\overline{P_3}$ for some $s\in[t]\setminus\{k,l\}$ and $G \setminus \{u_k,u_s,u_s'\}$ has a perfect matching $ (M_v\setminus \{u_ku_k',u_lu_l',u_su_s'\} ) \cup \{vu_l,u_k'u_l'\} $. Also, by symmetry, $u_lu_s'\in E(G)$ for all $s\in[t]\setminus\{k,l\}$.
\item $u_su_k'\in E(G)$ for all $s\in[t]\setminus\{k,l\}$, otherwise $G[\{v,u_s,u_k'\}]=\overline{P_3}$  for some $s\in[t]\setminus\{k,l\}$ and $G \setminus \{v,u_s,u_k'\}$ has a perfect matching $ (M_v\setminus \{u_ku_k',u_su_s'\} ) \cup \{u_ku_s'\} $. Also, by symmetry, $u_su_l'\in E(G)$ for all $s\in[t]\setminus\{k,l\}$.
\end{itemize}

Now $G[\{v,u_k,u_l'\}]=\overline{P_3}$  and $G \setminus \{v,u_k,u_l'\}$ has a perfect matching $ (M_v\setminus \{u_ku_k',u_su_s'\} ) \cup \{u_su_k',u_lu_s'\} $ for some $s\in[t]\setminus\{k,l\}$, a contradiction by Lemma \ref{lem: defn edg-hered equim}. It follows that $u_ku_l' \in E(G)$ for every pair $k,l \in [t]$, which completes the proof of the claim.
\end{cproof}\\

In the remaining of the proof, we consider two exclusive cases: $N_1'$ is independent or not. 

\begin{case}
 $N_1'$ is not independent.
\end{case}

We then claim that $N_1=\{u_1,u_2,\ldots,u_t\}$ is a nontrivial independent set which is complete to $G\sm N_1$. Indeed, we have $N_2 = \emptyset$  by Claim \ref{cl:N2 empty} and consequently, $t\geq 3$ since $G$ has at least 7 vertices. Thus, $N_1$ is an independent set complete to $N_1'$ by Claim   \ref{cl:complete} and complete to $\{v\}$ by definition. So, $N_1$ is complete to $N_1'\cup\{v\}=G\sm N_1$.

\begin{case}
 $N_1'$ is independent.
\end{case}

Recall that $N_1'$ is complete to $N_1$ by Claim \ref{cl:complete}. It can be observed that $N_2 \neq \emptyset$ since otherwise $G=K_{r,r+1}$ which is not factor-critical. If $N_2$ is complete to $N_1'$, then $S=N_1' \cup \{v\}$ is an independent set complete to $G \sm S$ as claimed. So, assume that $N_2$ is not complete to $N_1'$, that is, there exist $w \in N_2$ and $u_k' \in N_1'$ for $k \in [t]$ such that $wu_k' \notin E(G)$. Let without loss of generality $w=x_1$, that is $x_1u_k' \notin E(G)$.

Note that, for every $k \in [t]$,  $u_k'$ is adjacent to one of the endpoints of each edge in $\{x_1x_1',x_2x_2',\ldots,x_px_p'\}$; indeed if $\{x_i,x_i'\} \cap N(u_k')=\emptyset $ for some $i\in [p]$, then $G[\{u_k',x_i,x_i'\}]=\overline{P_3} $ and $G\setminus \{u_k',x_i,x_i'\} $ has a perfect matching $ (M_v\setminus \{x_ix_i',u_ku_k'\} ) \cup \{vu_k\} $, a contradiction. Let without loss of generality $ \{x_1',x_2',\ldots,x_p'\}$ be the neighbours of $u_k'$. We remark that $x_1$ must be adjacent to each vertex of $N_1$, otherwise, let $x_1u_l\notin E(G)$ for $l \in [t]$, then $G[\{x_1,u_l,u_k'\}]=\overline{P_3}$  and $G\setminus \{x_1,u_l,u_k'\} $ has a perfect matching as $( M_v\setminus \{x_1x_1',u_ku_k',u_lu_l'\} ) \cup \{vx_1',u_ku_l'\} $, a contradiction. Similarly, if there is a vertex $x_i$ for $i \in [p]$ such that $x_iu_k'\notin E(G)$, then $x_i$ is adjacent to each vertex of $N_1$.  \\
Now, we claim that $x_1,x_2,\ldots,x_p \notin N(u_k')$. Otherwise, let $x_iu_k' \in E(G)$ for $i \neq 1$. If there is a perfect matching $P$ in $G[\{u_k,x_1',x_i,x_i'\}]$, then consider the $\overline{P_3}$ induced by $\{v,x_1,u_k'\}$ and note that $G\setminus \{v,x_1,u_k'\} $ contains the perfect matching $( M_v\setminus \{x_1x_1',x_ix_i',u_ku_k'\} ) \cup P $, a contradiction. Assume $G[\{u_k,x_1',x_i,x_i'\}]$ has no perfect matching. Noting that $\nu(G[\{u_k,x_1',x_i,x_i'\}])=1$, this can only be the disjoint union of a star and (at most 2) isolated vertices, or a triangle and an isolated vertex, namely $K_{1,3},K_3 \cup K_1,P_3 \cup K_1,  P_2 \cup K_1 \cup K_1 $ and each of these graphs contains either $\overline{P_3}$ or an independent set $I$ of size $3$. For such $\overline{P_3}$'s and $I$'s, $G \setminus \overline{P_3}$ or $G\setminus I$ has a perfect matching  $( M_v\setminus \{x_1x_1',x_ix_i',u_ku_k'\} ) \cup \{vx_1,wu_k'\} $ where $w$ is the vertex remaining from $\{ u_k,x_1',x_i,x_i'\}$ after removing $\overline{P_3}$ or $I$ (note that $u_k'$ is adjacent to each of $\{ u_k,x_1',x_i,x_i'\}$). This contradicts being $\ESE$ or equimatchable.\\
Remind that for every $j \in [p]$, $x_j$ is complete to $N_1$. Moreover, we now claim that for every $j \in [p]$, $x_j$ must be complete to $N(u_k')$ (recall that $x_ju_k' \notin E(G)$), otherwise let $x_jx_i'\notin E(G)$, then $G[\{x_j,x_i',u_k'\}]=\overline{P_3}$ and   $G\setminus \{x_j,x_i',u_k'\} $ has a perfect matching $ ( M_v\setminus \{x_jx_j',x_ix_i',u_ku_k'\} ) \cup \{vx_j',x_iu_k\} $, a contradiction with $G$ being $\ESE$. So  $x_j$ must be complete to $N(u_k')\setminus N_1$. As a result, for every $j \in [p]$, $x_j$ is adjacent to each vertex of $N(u_k')$.  \\
Besides, for every $j \in [p]$, $x_j'$ has no neighbour in $ N_1$ since otherwise let $x_j'u_l \in E(G)$ then for some $k\in [p]$, $G[\{v,x_j,u_k'\}]$ induces a $\overline{P_3}$ such that  $G \setminus \{v,x_j,u_k'\}$ has a perfect matching $(M_v\setminus \{x_jx_j',u_lu_l',u_ku_k'\})  \cup \{u_lx_j',u_ku_l'\} $, a contradiction. \\  
Furthermore, for every $j \in [p]$, $x_j'$ is adjacent to each vertex of $ N_1'$ since otherwise let $x_j'u_l' \notin E(G)$ then $G[x_j',u_l',u_l]$ is a $\overline{P_3}$ such that  $G \setminus \{x_j',u_l',u_l\}$ has a perfect matching $(M_v\setminus \{x_jx_j',u_lu_l'\})  \cup \{vx_j\} $.  \\
Now, we will show that any two $x_i',x_j'$ can not be adjacent for $i,j \in [p]$. Assume the contrary, let $x_i'x_j' \in E(G) $ then $G[\{v,x_i,u_k'\}] \cong \overline{P_3}$ and  $G\setminus \{v,x_i,u_k'\}$ has a perfect matching  $ (M_v\setminus \{x_ix_i',x_jx_j',u_ku_k'\})  \cup \{u_kx_j,x_i'x_j'\} $, it gives a contradiction with being $\ESE$-graph. Hence, $\{u_1,u_2,\ldots,u_t\}\cup\{x_1',x_2',\ldots,x_p'\}$ is an independent set. On the other hand, $x_i$ is complete to $G \sm (N_1\cup N_2)$ for all $i \in [p]$. Hence, $S=\{u_1,u_2,\ldots,u_t\}\cup\{x_1',x_2',\ldots,x_p'\}$ is an independent set which is complete to $G\sm S$ as desired.
\end{proof}

For later purpose, we need to show that there is a nontrivial independent set $S$ complete to $G\sm S$ and having a special form with respect to a canonical decomposition of $G$.  

\begin{cor}\label{cor: structure of S clique N1' cup v}
Let $G$ be a factor-critical graph with at least 7 vertices. If $G$ is an $\ESE$-graph which is not an odd clique, then for some $v\in V(G)$ and some perfect matching $M_v$ of $G-v$, $G$ has a canonical decomposition $(N_1,N'_1,N_2)$ where $S=N_1' \cup \{v\}$ is a nontrivial independent set which is complete to $G\sm S=N_1\cup N_2$, and $N_1$ is an independent set.
\end{cor}

\begin{proof}
By Lemma \ref{lem: exist a clique S}, there is a nontrivial independent set $S$ complete to $G\sm S$.  Taking any vertex $v\in S$, since $G$ is factor-critical, $G-v$ has a perfect matching $M_v$. It is easy to see that $M_v$ matches the vertices of $S\sm v$ to a subset $S' \subset V(G\sm S)$, since $S$ is independent. Moreover, both $S\sm v$ and $S'$ are nonempty. We also remark that $S'$ is an independent set, since otherwise let  $y'z' \in G[S']$ and $yy',zz' \in M_v$ for $y,z \in S \sm v$ and $y',z' \in S'$, then $\{v,y,z\}$ is an independent set and  $( M_v\setminus \{yy',zz'\} ) \cup \{y'z'\} $ is a perfect matching on $G\sm \{v,y,z\}$, contradiction with equimatchability of $G$. \\
Now, one can observe that $G$ has a canonical decomposition where 
 $N_1'=S\sm v$, $N_1=S'$ and $N_2=G \sm (S\cup S')$. Note that, $N_1'$ is complete to $N_1\cup N_2$, and $N_1$ is an independent set.
\end{proof}

We define two graph families $\G_1$ and $\G_2$ corresponding to the cases where the nontrivial independent set $S$ described in Corollary \ref{cor: structure of S clique N1' cup v} has respectively two or more vertices and satisfying some additional properties due to being $\ESE$-graphs. We will show that all factor-critical $\ESE$-graphs (apart from odd cliques) fall into one of these two families. A graph $G$ belongs to $\G_1$ if  $G \cong K_{2r+1} \sm M$ for some nonempty matching $M$ and $r \geq 3$. A graph $G$ of $\G_1$ is illustrated in Figure \ref{fig: FC- ESE graph (a)} where the edges in $G{[}N_1 \cup N_2{]} \cong K_{2r-1} \sm (M \sm vu_1')$ are not drawn, and $S=\{v,u_1'\}$ is complete to $G \sm S$. Besides, $\G_2$ is defined as the family of graphs $G$ admitting an independent set $S$ of size at least $3$ which is complete to $G \sm S$ and such that  $\nu(G \setminus S)=1 $.  In Figure \ref{fig: FC- ESE graph (b)},  we show an illustration of a graph $G$ in $\G_2$ where $S=N_1' \cup \{v\}$ with $\vert S \vert \geq 3$ and $\nu(G \sm S)=1$. Again, the edges in $G{[}N_1 \cup N_2{]}$ are not drawn but just described by the property $\nu(G \sm S)=1$. \\

\begin{figure}[htb]
\centering     
\subfigure[Illustration of a graph $G$ in $\G_1$ with $S=\{v,u'_1\}$ and $G \cong K_{2r+1} \sm M$ where $M$ is a matching containing $vu_1'$.]{\label{fig: FC- ESE graph (a)}
\begin{tikzpicture}[scale=.7]
\node [noddee] at (2,1.5) (u1) [label=above:$u_1$]  {};
\node [noddee] at (2,-.5) (x2)  {};
\node [noddee] at (2,-1.5) (x22)  {}
	edge [ultra thick, red] (x2);
\node at (2,-1.85) (dot) {$\vdots$};
\node [noddee] at (2,-2.5) (xp)  {};
\node [noddee] at (2,-3.5) (xpp)  {}
	edge [ultra thick, red] (xp);
\node [noddee] at (0,0) (v) [label=left:$v$]  {}
	edge [] (u1)
	edge [] (x2)
	edge [] (x22)
	edge [] (xp)	
	edge [] (xpp);
\node [noddee] at (4,1.5) [label=above:$u_1'$] (u11)  {}
	edge [ultra thick, red] (u1)
	edge [] (x2)
	edge [] (x22)
	edge [] (xp)	
	edge [] (xpp);
\draw[thick] (2,1.5) ellipse (.7cm and 1cm);	
\draw[thick](4,1.5) ellipse (.7cm and 1cm);	
\draw[thick] (2,-2) ellipse (.8cm and 2cm) ;		
\node at (1.8,1.7) (N1) [label=above left:$N_1$] {};		
\node at (4.2,1.7) (N11) [label=above right:$N_1'$] {};			
\node at (2,-3.5) (N2) [label=below right:$N_2$] {};
\node at (-1,2) {};
\node at (6,2) {};
\end{tikzpicture}  
}
\hspace*{2cm}
\subfigure[Illustration of a graph $G$ in $\G_2$ with $S=N_1' \cup \{v\}$ and $\nu(G\sm S)=1$.]{\label{fig: FC- ESE graph (b)}
\begin{tikzpicture}[scale=.7]
\node [noddee] at (2,3) (u1)  {};
\node [noddee] at (2,2) (u2)  {};
\node at (2,1.3) (dot) {$\vdots$};
\node [noddee] at (2,.5) (ut)  {};
\node [noddee] at (2,-1) (x1)  {};
\node [noddee] at (2,-2) (x11)  {}
	edge [ultra thick, red] (x1);	
\node [noddee] at (0,0) (v) [label=left:$v$]  {}
	edge [] (u1)
	edge [] (u2)
	edge [] (ut)
	edge [] (x1)
	edge [] (x11);
\node [noddee] at (4,3) (u11)  {}
	edge [ultra thick, red] (u1)
	edge [] (u2)
	edge [] (ut)
	edge [] (x1)
	edge [] (x11);
\node [noddee] at (4,2) (u22)  {}
	edge [] (u1)
	edge [ultra thick, red] (u2)
	edge [] (ut)
	edge [] (x1)
	edge [] (x11);
\node at (4,1.3) (dot) {$\vdots$};
\node [noddee] at (4,.5) (utt)  {}
	edge [] (u1)
	edge [ultra thick, red] (ut)
	edge [] (u2)
	edge [] (x1)
	edge [] (x11);
\draw[thick] (2,1.8) ellipse (.8cm and 1.8cm);	
\draw[thick](4,1.8) ellipse (.8cm and 1.8cm);	
\draw[thick] (2,-1.4) ellipse (.75cm and 1cm);		
\node at (2,3) (N1) [label=above left:$N_1$] {};		
\node at (4,3) (N11) [label=above right:$N_1'$] {};			
\node at (2,-2) (N2) [label=below right:$N_2$] {};	
\node at (-1,2) {};
\node at (5.5,2) {};		
\end{tikzpicture} 
}
\caption{Factor-critical $\ESE$-graph families $\G_1$ and $\G_2$ where the bold (red) edges illustrate a perfect matching $M_v$ of $G\sm v$ which defines the canonical decomposition $(N_1, N'_1,N_2)$.}
\label{fig:FC- ESE graph}
\end{figure}
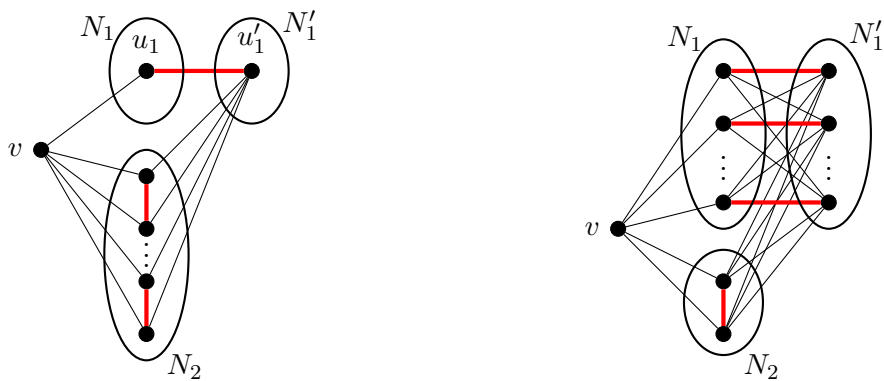

\setcounter{claim}{0}

\begin{thm} \label{thm: chracterization of edge-stable}
Let $G$ be a factor-critical graph with at least $7$ vertices. Then,
$G$ is $\ESE$  if and only if  $G$ is either an odd clique or a graph in $\G_1 \cup \G_2$.
\end{thm}

\begin{proof}
Let $G$ be a factor-critical graph with at least $7$ vertices. It is clear that $G \cong K_{2r+1}$ is an $\ESE$-graph for $r\geq 3$. Assume that $G$ belongs to $\G_1$ or $\G_2$ and we will show that $G$ is again $\ESE$. First, let $G$ be in $\G_1$, that is $G \cong K_{2r+1} \setminus M$ where $M$ is a nonempty matching, then every maximal matching of $G$ has $r$ edges (indeed $G$ has no independent set of size 3, hence, no maximal matching of size $r-1$), thus $G$ is equimatchable. In addition, $G$ is $\overline{P_3}$-free. Hence, $G$ is $\ESE$ by Lemma \ref{lem: defn edg-hered equim}. 

Now, let $G$ be in $\G_2$, then by definition of $\G_2 $, we have $\vert S \vert\geq 3$ and  $G \setminus S $ induces a graph whose matching number is equal to 1. Consider a vertex $v\in S$. Since $G$ is factor-critical, $G- v$ has a perfect matching where all vertices in $S$ are necessarily matched to some vertex in $G \setminus S$. If this perfect matching contains an edge in $G\setminus S$ then $|S|=|V(G\setminus S)|-1$, otherwise $|S|=|V(G\setminus S)|+1$. Applying the same argument to a vertex $v'\in G\setminus S$, we have either $|S|=|V(G\setminus S)|-3$ (in case a perfect matching of $G-v'$ contains an edge in $G\setminus S$), or $|S|=|V(G\setminus S)|-1$ (otherwise). It follows that the only possible case is $\vert V(G\setminus S) \vert=\vert S \vert +1$.  

 We claim that $G$ is equimatchable. Since $S$ is complete to $G \setminus S$,  we have either $I \subseteq S$ or $I \subseteq G \setminus S$ for any independent set $I$ with $3$ vertices. If $I \subseteq S$, then $ G  \setminus I$ has no perfect matching; indeed all vertices of $S \setminus I $ have to be matched to vertices in $G \setminus S$, leaving $4$ vertices in $G \setminus S$ exposed (because $\vert I \vert=3$) that can not be all saturated since $\nu(G \setminus S)=1$. Otherwise, if $I \subseteq G \sm S$, then $G \sm I$ has no perfect matching because there are $\vert S \vert$ vertices that can not be matched to remaining $\vert S \vert -2$  vertices of $(G \sm S) \sm I$. Hence, by Lemma \ref{lem: defn equim}, $G$ is equimatchable. Now, we claim that $G$ is edge-stable. Let $  R \subset V(G)$ such that $G[R]\cong \overline{P_3}$.  It then follows that $  R \subseteq V(G \sm S)$ since every vertex of $S$ is adjacent to all vertices in $G \sm S$. It can be observed that there is no perfect matching in $G \sm R$ because there are $\vert S \vert$ vertices that can not be matched to remaining $\vert S \vert -2$  vertices of $(G \sm S) \sm R$.  Hence, by Lemma \ref{lem: defn edg-hered equim}, $G$ is edge-stable.\\

\noindent
Now, let us show the converse. Assume that $G$ is a factor-critical $\ESE$-graph with $\vert V(G) \vert=2r+1$ for $r\geq 3$. We will show that if $G$ is not an odd clique then it is either in $\G_1$ or in $\G_2$. By Corollary \ref{cor: structure of S clique N1' cup v}, for some $v\in V(G)$ and some perfect matching $M_v$ of $G-v$, $G$ has a canonical decomposition $(N_1,N'_1,N_2)$ where $S=N_1' \cup \{v\}= \{v,u_1',u_2',\ldots,u_t'\}$ is a nontrivial independent set which is complete to $G\sm S$, and $N_1$ is a nonempty independent set (see Figure \ref{fig:minimal matching M_v isolating v}). \\
First, let $\vert S \vert=2$, then $\vert N_1 \vert=1$. We note that in this case $p\geq 2$ since $2r+1 \geq 7$. We will show that $G\cong K_{2r+1}\setminus M$, where $M$ is a matching of $K_{2r+1}$ containing $vu'_1$.
We suppose that it is not true, then, since  $G\not \cong K_{2r+1}$, $G$ is obtained from $K_{2r+1}$ by removing an edge set which do not form a matching. This means in particular that there is a vertex $w_1$ with at least two non-neighbors $w_2$ and $w_3$. In addition, we know by Corollary \ref{cor: structure of S clique N1' cup v} that $\{v,u_1'\}$ is complete to $N_1 \cup N_2$, implying that the missing edges are between vertices in $N_1\cup N_2$. It follows that there exists $w_1,w_2,w_3 \in N_1 \cup N_2$ such that $G[\{w_1,w_2,w_3\}] \cong \overline{P_3}$ or an independent set of size 3.   
Noting that $S$ is complete to $N_1\cup N_2$ and that without loss of generality,  $x_1,x_1'$ can be considered for some $x_i,x_i'$ pair, the following cases cover all possibilities for $w_1,w_2,w_3$.
\begin{itemize}
\item[$(a)$] If $w_1=u_1 \in N_1$, $w_2=x_1 \in N_2$, $w_3=x_1' \in N_2$, then $P=\{vx_2,x_2'u_1',x_3x_3',\ldots x_px_p'\}$ is a perfect matching in $G \setminus \{w_1,w_2,w_3\}$.
\item[$(b)$] If $w_1=u_1 \in N_1$, $w_2=x_1 \in N_2$, $w_3=x_2 \in N_2$, then $P=\{vx_1',x_2'u_1',x_3x_3',\ldots,x_px_p'\}$ is a perfect matching in $G \setminus \{w_1,w_2,w_3\}$.

\item[$(c)$]If $w_1=x_1$, $w_2=x_2$, $w_3=x_2'$, then  $P=\{u_1u_1',vx_1',x_3x_3',\ldots,x_px_p'\}$ is a perfect matching in $G \setminus \{w_1,w_2,w_3\}$.
\item[$(d)$] If $w_1=x_1$, $w_2=x_2$, $w_3=x_3$, then there are two cases. If there is an edge $e \in E(G[\{x_1',x_2',x_3'\}])$, say without loss of generality $e=x_1'x_2'$, then  $P=\{u_1u_1',e=x_1'x_2',vx_3',x_4x_4',\ldots,x_px_p'\}$ is a perfect matching in $G \setminus \{w_1,w_2,w_3\}$.
Assume now that $G[\{x_1',x_2',x_3'\}]$ is a null graph. If one of the edges $u_1x_1'$, $u_1x_2'$ or $u_1x_3'$ exists, say without loss of generality $u_1x_1'$, then $G \setminus \{w_1,w_2,w_3\}$ has a perfect matching $P=\{vx_2',u_1x_1',u_1'x_3',x_4x_4',\ldots,x_px_p'\}$. Otherwise, we can conclude in exactly the same manner as in Case $(b)$ by considering $\{u_1,x_1,x_2\}$ as $ \{w_1,w_2,w_3\}$.
\end{itemize} 
In all these cases, we conclude by Lemmas \ref{lem: defn equim} and \ref{lem: defn edg-hered equim} that there is a contradiction, hence, $G \cong K_{2r+1}\setminus M$ where $M$ is a matching of $K_{2r+1}$ containing $vu'_1$ and consequently $G$ belongs to $\G_1$ (see Figure \ref{fig: FC- ESE graph (a)}).\medskip

\noindent
Assume now $\vert S \vert \geq 3$. We need the following claim to show that $G[N_1 \cup N_2]$ induces a graph whose matching number is equal to 1. 

\begin{claim}
$\vert N_2 \vert =2$
\end{claim}
\begin{cproof}
First, $\vert N_2 \vert\neq 0$, since otherwise for $w \in N_1$, $G- w$ has no perfect matching (note that $S=\{v\}\cup N_1'$ is an independent set of $G-v$ of cardinality two more than $(G-v)\sm S$), hence, $G$ is not factor-critical. Now, assume for a contradiction that $|N_2|>2$, that is $p > 1$. Note that $\vert N_1 \vert =\vert N_1' \vert =t  \geq 2$ due to $\vert S \vert \geq 3 $, let $u_i',u_j' \in N_1'$. First we observe that no vertex in $N_2$ forms an independent set with $\{u_i,u_j\}$. Assume for a contradiction that there exists a vertex $w \in N_2$, say without loss of generality $w=x_1$ such that $\{u_i,u_j\} \cap N(x_1)=\emptyset$, that is $\{u_i,u_j,x_1\}$ is an independent set in $G$, then we can obtain a perfect matching $ (M_v\setminus \{x_1x_1',x_2x_2',u_iu_i',u_ju_j'\})  \cup \{vx_1',u_i'x_2,u_j'x_2'\} $, a contradiction with being equimatchable. It follows that for all $w \in N_2$, $N(w) \cap \{u_i,u_j\} \neq \emptyset$. If there exists $u_i \in N_1$ such that $u_i$ is not adjacent to vertices $\{x_k,x_k'\}$, say without loss of generality $k=1$, then $\{u_i,x_1,x_1'\}$ induces a $\overline{P_3}$ and $G\setminus \{u_i,x_1,x_1'\}$ has a perfect matching $ (M_v\setminus \{x_1x_1',x_2x_2',u_iu_i',u_ju_j'\})  \cup \{vu_j,u_i'x_2,u_j'x_2'\} $, contradiction with being edge-stable. 
Consequently, for any pair $u_i,u_j \in N_1$ and $ x_k,x_k' \in N_2$, the graph induced by $\{ u_i,u_j,x_k,x_k'\}$ contains a perfect matching $P$ since the only graphs on $4$ vertices with matching number $1$ are $K_{1,3}$, $K_3 \cup K_1$ and $K_2 \cup K_1 \cup K_1$, and 
$G[\{u_i,u_j,x_k,x_k'\}]$ induces none of them by the above properties. Now, we notice that  $\{v,u_i',u_j'\}$ is an independent set and $G\setminus \{v,u_i',u_j'\}$ has a perfect matching $ (M_v\setminus \{x_kx_k',u_iu_i',u_ju_j'\})  \cup P $.  This is a contradiction with being equimatchable. So, $\vert N_2 \vert=2$. 
\end{cproof}\\

\noindent
Let $N_2=\{x_1,x_1'\}$. Now, if $\nu(G \setminus S) \geq 2$, then there are $4$ vertices $u_i,u_j,x_1,x_1'$ such that $G[\{u_i,u_j,x_1,x_1'\}]$ has a perfect matching $P$ (remind that $u_iu_j \notin E(G)$ since $N_1$ is an independent set); now $(M_v\setminus \{x_1x_1',u_iu_i',u_ju_j'\})  \cup P $ is a maximal matching in $G$ of size $(|V(G)|-3)/2$, contradiction with equimatchability (see Figure \ref{fig: FC- ESE graph (b)}). Hence, $G$ belongs to $\G_2$. This completes the proof.
\end{proof}


\begin{cor}\label{cor: finitely many FC ESE}
For every $r\geq 3$, there are exactly $2r+2$ factor-critical $\ESE$-graphs on $2r+1$ vertices.
\end{cor}
\begin{proof}
A nonempty matching of $K_{2r+1}$ has size between $1$ and $r$, implying that there are $r$ non-isomorphic graphs of family $\G_1$  in Theorem \ref{thm: chracterization of edge-stable}, that is  isomorphic to $K_{2r+1}\sm M$ for some nonempty matching $M$. If $G$ is of family $\G_2$ (see Figure \ref{fig: FC- ESE graph (b)}) then,  $\nu(G \sm S)=1$ implies that $G[N_1\cup N_2]$ can only be a disjoint union of isolated vertices and one triangle or one star where the edge $x_1x_1'$ belongs to this unique triangle or star. Since all vertices of $N_1\cup N_2$ are symmetric with respect to their neighborhoods outside of $N_1\cup N_2$, there are exactly $r$ ways of forming a star (note that $N_1$ has $(2r+1-3)/2 = r-1$ vertices) and just one way to form a triangle. Summing up all possibilities together being $G \cong K_{2r+1}$, there are in total $2r+2$ factor-critical $\ESE$-graphs.
\end{proof}

Next, we determined all factor-critical $\ESE$-graphs whose orders are at most $5$ by using computer programs written in Python-Sage. 

\begin{rem}\label{rem: small FC ESE}
There are only $6$ factor-critical $\ESE$-graphs on at most $5$ vertices:  $K_3$, $K_5$, $C_5$ and the three graphs in Figure \ref{fig: some ESE-graphs size of 5 }.
\end{rem}

\begin{figure}[htb]
\centering     
\subfigure{\label{fig:a}
\begin{tikzpicture}[scale=.8]
\node [noddee] at (-.25,0) (x1)  {};
\node [noddee] at (-1,1.1) (x2)  {}
	edge [] (x1);
\node [noddee] at (.5,2) (x3)  {}
	edge [] (x1)
	edge [] (x2);
\node [noddee] at (2,1.1) (x4)  {}
	edge [] (x1)	
	edge [] (x2);
\node [noddee] at (1.25,0) (x5)  {}
	edge [] (x1)
	edge [] (x2)
	edge [] (x3)
	edge [] (x4);
\end{tikzpicture} 
}
\hspace*{1cm}
\subfigure{\label{fig:b}
\begin{tikzpicture}[scale=.8]
\node [noddee] at (-.25,0) (x1)  {};
\node [noddee] at (-1,1.1) (x2)  {}
	edge [] (x1);
\node [noddee] at (.5,2) (x3)  {}
	edge [] (x1)
	edge [] (x2);
\node [noddee] at (2,1.1) (x4)  {}
	edge [] (x1)	
	edge [] (x2);
\node [noddee] at (1.25,0) (x5)  {}
	edge [] (x2)
	edge [] (x3)
	edge [] (x4);
\end{tikzpicture}
}
\hspace*{1cm}
\subfigure{\label{fig:c}
\begin{tikzpicture}[scale=.8]
\node [noddee] at (-.25,0) (x1)  {};
\node [noddee] at (-1,1.1) (x2)  {}
	edge [] (x1);
\node [noddee] at (.5,2) (x3)  {}
	edge [] (x2);
\node [noddee] at (2,1.1) (x4)  {}
	edge [] (x1)
	edge [] (x3);
\node [noddee] at (1.25,0) (x5)  {}
	edge [] (x1)
	edge [] (x2)
	edge [] (x4);
\end{tikzpicture}  
}

\caption{Some factor-critical $\ESE$-graphs with $5$ vertices.}
\label{fig: some ESE-graphs size of 5 }
\end{figure}
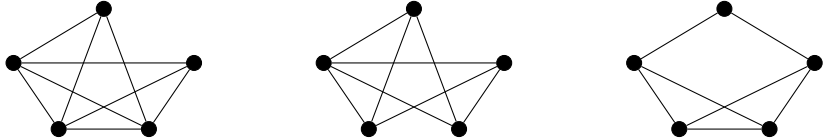

By observing that all the graphs described in Theorem \ref{thm: chracterization of edge-stable} and Remark \ref{rem: small FC ESE} are 2-connected, we obtain as a byproduct the following.

\begin{cor}\label{cor: FC ESE are 2conn}
Factor-critical $\ESE$-graphs are 2-connected.
\end{cor}


\section{ESE-graphs with a cut vertex} \label{sec: ESE-graphs with a cut vertex}

The main objective of this section is to show that $\ESE$-graphs with a cut vertex are bipartite. Then we will complete the characterization of all $\ESE$-graphs in the next section with bipartite $\ESE$-graphs.

Recall that we consider only connected $\ESE$-graph. We will first extend the following known result to $\ESE$-graphs.

\begin{lem}\label{lem:cutvertex}\cite{eqm_regular}
Let $G$ be a connected equimatchable graph with a cut vertex $v$, then each connected component of $G-v$ is also equimatchable.
\end{lem}

\begin{lem}\label{lem: cut vertex v, connected comp of G-v}
Let $G$ be a connected $\ESE$-graph with a cut vertex $v$, then each connected component of $G-v$ is also $\ESE$.
\end{lem}
\begin{proof}
Let $v$ be a cut vertex and $H_1,H_2,\ldots,H_k$ $(k\geq 2)$ be the connected components of $G-v$. Assume $H_i$ is not $\ESE$ for some $i\in [k]$. However, we know from Lemma \ref{lem:cutvertex} that it is equimatchable. Then, there are two maximal matchings $M_1$ and $M_2$ of different sizes in $H_i\sm w_1w_2$ for some $w_1w_2\in E(H_i)$. Let $M$ be a maximal matching of $G\sm H_i$ containing $uv$ where $u\in H_j$ for some $j\neq i$. Then $M_1'=M_1\cup M$ and $M_2'=M_2\cup M$ are maximal matchings of $G \sm \{w_1w_2\}$ with different sizes, contradiction with equimatchability of $G \sm \{w_1w_2\}$.
\end{proof}

The following well-known structural result on maximum matchings will guide us towards our objective.

\begin{thm}[Gallai-Edmonds decomposition] \cite{LP} \label{thm: Gallai} Let $G$ be a graph, $D(G)$ the set of vertices of $G$ that are not saturated by at least one maximum matching, $A(G)$ the set of vertices of $V(G) \sm D(G)$ with at least one neighbor in $D(G)$, and $C(G) \stackrel{def}{=} V(G) \sm (D(G) \cup A(G)) $. Then:
\begin{itemize}
\item[$(i)$] the connected components of $G[D(G)]$ are factor-critical,\medskip
\item[$(ii)$] $G[C(G)]$ has a perfect matching,\medskip
\item[$(iii)$] every maximum matching of $G$ matches every vertex of $A(G)$ to a vertex in a distinct component of $G[D(G)]$.
\end{itemize}
\end{thm}

Item (iii) of Theorem \ref{thm: Gallai} means that in any maximum matching of $G$, different vertices of $A(G)$ never match to vertices in the same component of $G[D(G)]$. The following lemma states that if a graph is equimatchable (but not randomly matchable), then its Gallai-Edmonds decomposition admits some additional properties. Indeed, Lemma \ref{lem: Gallai- C(G) empty and A(G) is ind. set} holds for $\ESE$-graphs with a cut vertex (clearly, $K_2$ has no cut vertex) since they are equimatchable and they do not admit a perfect matching by Corollary \ref{rem:unchanged matching number}.

\begin{lem}\cite{Plummer}\label{lem: Gallai- C(G) empty and A(G) is ind. set}
Let $G$ be a connected equimatchable graph with no perfect matching. Then the sets $C(G)$ and $A(G)$ as defined in the Gallai-Edmonds decomposition of $G$ are such that $C(G)= \emptyset$ and $A(G)$ is an independent set of $G$.
\end{lem}

In the sequel, we show that if we use the property of being edge-stable in addition to equimatchability,  each component of $G[D(G)]$ in the Gallai-Edmonds decomposition of an $\ESE$-graph $G$ with a cut vertex restricts to a single vertex. We will show this in two steps.

\begin{lem}\label{lem:ESE with cut vertex, G-v has no FC}
Let $G$ be a connected $\ESE$-graph and $v\in A(G)$ be a cut vertex where $A(G)$ is defined as in the Gallai-Edmonds decomposition of $G$. Then every factor-critical component of $G-v$ is $K_1$.
\end{lem}
\begin{proof}
Let $F$ be a factor-critical component of $G-v$ which is not $K_1$. Since $F$ is factor-critical, it has an odd number of vertices. Thus $F$ has at least  $3$  vertices. 

We first note that since $G$ is $\ESE$ and $F$ is a connected component of $G-v$, Lemma~\ref{lem: cut vertex v, connected comp of G-v} implies that $F$ is $\ESE$. Consider $w\in N_F(v)$, we claim that $F-w$ is also $\ESE$. Assume for a contradiction that $F-w$ is not $\ESE$. If $F-w$ is not equimatchable, let $M_1$ and $M_2$ be two maximal matchings of different sizes in $F-w$. Let also $M$ be a maximal matching of $G\sm (F\cup \{v\})$. Then $M\cup M_1\cup \{vw\}$ and  $M\cup M_2\cup \{vw\}$ are two maximal matchings  of different sizes in $G$, contradicting that $G$ is equimatchable. Now, if $F-w$ is equimatchable but not $\ESE$, then let $M_1$ and $M_2$ be two maximal matchings of different sizes in $(F-w)\sm e$ where $e\in E(F-w)$.  Then $M\cup M_1\cup \{vw\}$ and  $M\cup M_2\cup \{vw\}$ are two maximal matchings  of $G\sm e$ having different sizes, contradicting that $G$ is $\ESE$.

Since $F$ is factor-critical $\ESE$ (with at least 3 vertices), by Corollary \ref{cor: FC ESE are 2conn}, $F$ is 2-connected. Therefore $F-w$ is a connected $\ESE$-graph which admits a perfect matching. This contradicts Corollary \ref{rem:unchanged matching number} unless $F-w$ is a $K_2$. So, assume $F-w$ is a $K_2$, and consequently, $F$ is a $K_3$ (since $K_3$ is the only factor-critical graph on 3 vertices). Let $M$ be a maximal matching of $G\sm (F\cup\{v\})$. Let $w_1w_2$ be the edge in $F-w$. Then $M\cup \{vw,w_1w_2\}$ is a maximal matching of $G$, which is also of maximum size since $G$ is equimatchable. On the other hand, $M\cup \{vw\}$ is a maximal matching in $G\sm e$ where $e=w_1w_2$; this matching is also of maximum size (in $G\sm e$) since $G\sm e$ is equimatchable. It follows that $\nu(G\sm e)=\nu(G)-1$ which contradicts that $G$ is $\ESE$ by Proposition \ref{prop:unchanged matching number}. This concludes the proof.
\end{proof}


\begin{lem}\label{lem: ESE cut vertex-- componnet is a vertex}
Let $G$ be a connected $\ESE$-graph with a cut vertex. Then every component of $D(G)$ in the Gallai-Edmonds decomposition is $K_1$.
\end{lem}
\begin{proof}
Consider a connected $\ESE$-graph $G$ with a cut vertex. By Corollary \ref{cor: FC ESE are 2conn}, $G$ is not factor-critical. Moreover, $G$ has at least 3 vertices since it has a cut vertex; thus, by Corollary \ref{rem:unchanged matching number}, $G$ does not have a perfect matching. Consequently, $G$ has a Gallai-Edmonds decomposition as described in Lemma \ref{lem: Gallai- C(G) empty and A(G) is ind. set} where $C(G)=\emptyset$ nd $A(G)$ is a nonempty independent set. 

Let $F$ be a factor-critical component of $G[D(G)]$ which is not $K_1$. Since $F$ is factor-critical, it has an odd number of vertices. Thus $F$ has at least  $3$  vertices. By Theorem \ref{thm: Gallai} $(iii)$, and the equimatchability of $G$, every maximal matching of $G$ matches every vertex of $A(G)$ to a vertex in a distinct component of $G[D(G)]$. This implies that for vertices $w_1,w_2 $ belonging to the same connected component of $G[D(G)]$ and $a_1,a_2 \in A(G)$, if $w_1a_1\in E(G)$ and $w_2a_2\in E(G)$, then we have $w_1=w_2$ or $a_1=a_2$. It follows that, for every edge $wa$ where $w\in V(F)$ and $a\in A(G)$, at least one of $w$ and $a$ is a cut vertex. 

If $a$ is a cut vertex then $F$ is a connected component of $G-a$ which is factor-critical by Theorem \ref{thm: Gallai} $(i)$; however $F$ is different from $K_1$, contradiction by Lemma \ref{lem:ESE with cut vertex, G-v has no FC}. Otherwise, $w$ is a cut vertex, and every connected component of $G-w$, in particular $F-w$, is also $\ESE$ by Lemma \ref{lem: cut vertex v, connected comp of G-v}. Since $F$ is factor-critical, $F-w$ has a perfect matching; this contradicts that $F-w$ is $\ESE$ by Corollary \ref{rem:unchanged matching number} unless every connected component of $F-w$ is a $K_2$. So, assume that every connected component of $F-w$ is a $K_2$. Let $w_1w_2$ be an edge  in $F-w$ and $M$ be a maximal matching of $G\sm \{a,w,w_1,w_2\}$.  Then $M\cup \{wa,w_1w_2\}$ is a maximal matching of $G$ which is also maximum since $G$ is equimatchable. Besides, $M\cup\{wa\}$ is a maximal matching of $G\sm w_1w_2$ (note that $w_1$ and $w_2$ are only adjacent to $w$  in $G \sm w_1w_2$) which is also maximum since $G\sm w_1w_2$ is equimatchable. This contradicts that $G$ is $\ESE$ by Proposition \ref{prop:unchanged matching number} which states that the removal of an edge whose endpoints do not form a connected component does not change the matching number of an $\ESE$-graph.
\end{proof}

We obtain the main result of this section, namely Theorem \ref{thm: ESE with cut vertex is bip}, by combining Lemma \ref{lem: Gallai- C(G) empty and A(G) is ind. set} and Lemma \ref{lem: ESE cut vertex-- componnet is a vertex}. Lemma  \ref{lem: Gallai- C(G) empty and A(G) is ind. set} implies that an $\ESE$-graph $G$ has a Gallai-Edmonds decomposition where $C(G)=\emptyset$ and $A(G)$ is an independent set. Lemma  \ref{lem: ESE cut vertex-- componnet is a vertex} shows that if in addition $G$ has a cut vertex, then every component in $G[D(G)]$ is a $K_1$. This structure clearly implies a bipartite graph.

\begin{thm}\label{thm: ESE with cut vertex is bip}
Connected $\ESE$-graphs with a cut vertex are bipartite.
\end{thm}

\section{Bipartite ESE-graphs}\label{sec: bip ESE}

Having characterized all 2-connected factor-critical $\ESE$-graphs in Section \ref{sec:FC ESE} and having shown that $\ESE$-graphs with a cut vertex are bipartite (Theorem \ref{thm: ESE with cut vertex is bip}), we now consider bipartite $\ESE$-graphs to complete the characterization of all $\ESE$-graphs. We will see that bipartite $\ESE$-graphs can be characterized in a way very similar to bipartite equimatchable graphs.

\begin{lem}\cite{Plummer} \label{Lem: chr of equim bip graph}
A connected bipartite graph $G=(U \cup W, E)$ with $\vert U \vert \leq \vert W \vert$ is equimatchable if and only if for every $u \in U$, there exists $S\subseteq N(u)$ such that $S \neq \emptyset$ and $\vert N(S) \vert \leq \vert S \vert $. 
\end{lem}

Let us remind the well-known Hall's Theorem in order to obtain a more intuitive reformulation of Lemma \ref{Lem: chr of equim bip graph}: in a bipartite graph $G=(X\cup Y, E)$ with $|X| \leq |Y|$, there exists a matching saturating all vertices in $X$ if and only if for all subset $A\subseteq  X$, we have $|N(A)|\geq |A|$. The contrapositive of Lemma \ref{Lem: chr of equim bip graph} states that  $G=(U \cup W, E)$ with $\vert U \vert \leq \vert W \vert$ is not equimatchable if and only if there is a vertex $u\in U$ such that for all nonempty $S\subseteq N(u)$, we have $|N(S)|>|S|$. Hall's condition applied to the bipartite subgraph induced by $N(u)$ and $N(N(u))$ implies that the later condition is equivalent to the fact that there is a vertex $u\in U$ such that there is a matching of $G-u$ saturating all vertices of $N(u)$, or alternatively leaving $u$ exposed. The contrapositive of this equivalence suggests the following more intuitive statement.

Let us call a vertex $v \in V(G)$ \emph{strong} (in $G$) if every maximal matching of $G$ saturates $v$, or equivalently, if $v$ is not isolated vertex, there is no maximal matching of $G - v$ saturating all neighbors of $v$.

\begin{cor}\label{lem:new lemma}
Let $G=(U \cup W, E)$ be a connected bipartite graph with $\vert U \vert \leq \vert W \vert$. Then $G$ is equimatchable if and only if every vertex in $U$ is strong (or equivalently every maximal matching of $G$ saturates $U$).
\end{cor}

Next we point out a remark which allows us to restrict our attention to bipartite graphs with one part of the bipartition having strictly less vertices than the other part.

\begin{rem}\label{rem:U less than W}
Let  $G=(U \cup W, E)$, $\vert U \vert \leq \vert W \vert$ be a connected bipartite $\ESE$-graph with at least three vertices, then $|U| < |W|$.
\end{rem}
\begin{proof}
By Lemma \ref{Lem: chr of equim bip graph}, every maximal matching of $G$ saturates $U$. It follows that if $|U|=|W|$ then $G$ has a perfect matching, contradiction with Corollary \ref{rem:unchanged matching number}.
\end{proof}

It should be noted that in addition to connected bipartite $\ESE$-graphs with at least three vertices and  $|U| < |W|$, we also have $K_1$ and $K_2$ which are obviously bipartite $\ESE$ (by convention).

A strong vertex $u$ with $d(u) \geq 2$ is called  \emph{square-strong} if for every $v \in N(u) $, $u$ is strong in $G-v$. It follows from this definition that if $u$ is a square-strong vertex, then for all $v\in N(u)$, every maximal matching of $G-v$ saturates $u$, or equivalently, every maximal matching of $G$ leaves at least one vertex of $N(u)$ exposed.

\begin{prop}\label{prop: charc of ESE bip graph}
Let $G=(U \cup W, E)$ be a connected bipartite graph with $\vert U \vert < \vert W \vert$. Then the followings are equivalent. 
\begin{itemize}
\item[$(i)$] $G$ is $\ESE$.
\item[$(ii)$] Every vertex of $U$ is square-strong.
\item[$(iii)$] For every $u \in U$, there exists a nonempty set $S \subseteq N(u)$ such that $\vert N(S) \vert \leq \vert S \vert -1 $. 
\end{itemize}
\end{prop}
\begin{proof}

$(i) \Rightarrow (ii):$ Given a connected $\ESE$ graph $G=(U \cup W, E)$ with $\vert U \vert < \vert W \vert$. Since $G$ is equimatchable, by Corollary \ref{lem:new lemma}, every vertex in $U$ is strong, and thus,  of degree at least $2$. Suppose by contradiction that there exists $u \in U$ such that $u$ is not square-strong. This means that there is a vertex $v \in N(u)$ such that $u$ is not strong in $G-v$, that is, there is a maximal matching of $G-v$ which leaves $u$ exposed; let $M'$ be such a maximal matching of $G-v$. Then $M'\cup \{uv\}$ is a maximal matching of $G$ of size one more than $M'$, contradiction with $G$ being $\ESE$ by  Proposition \ref{prop:unchanged matching number}.

$(ii) \Rightarrow (iii):$ Assume every vertex of $U$ is square-strong. It means that for each $u \in U$, no matching in $G$ saturates $N(u)$. Therefore, Hall's condition does not hold for $N(u)$, and consequently there exists a nonempty set $S \subseteq N(u)$ such that $\vert N(S) \vert < \vert S \vert  $.

$(iii) \Rightarrow (i):$ Remark that $G$ is equimatchable  by Lemma \ref{Lem: chr of equim bip graph}. It remains to show that $G \setminus e$ is equimatchable for every $e=w_1w_2 \in E(G)$. Consider a vertex $u \in V(G\setminus e)$, then by assumption, there exists a nonempty set $S \subseteq N(u)$ such that  $\vert N(S) \vert \leq \vert S \vert-1 $.  If the endpoints of $e$ belong to $S\cup\{u\}$, say $w_1 \in S$ and $w_2 = u$ , let $S'=S\setminus w_1$, then we have $ \vert N_{G \setminus e}(S') \vert \leq \vert N_{G}(S') \vert \leq \vert N_G(S) \vert \leq \vert S \vert -1 =\vert S' \vert$ and $S' \subseteq N_{G \setminus e}(u)$, which implies by Lemma \ref{Lem: chr of equim bip graph} that $G \setminus e $ is equimatchable. For the other cases, we conclude similarly since $  \vert N_{G \setminus e}(S) \vert \leq \vert N_G(S) \vert \leq \vert S \vert$. As a result, $G\setminus e$ is equimatchable. 
\end{proof}
One can reformulate Proposition \ref{prop: charc of ESE bip graph} as follows: A connected bipartite graph $G=(U \cup W, E)$ with $\vert U \vert < \vert W \vert$ is $\ESE$ if and only if for every $u\in U$, every maximal matching of $G$ leaves at least one vertex of $N(u)$ exposed. 
Now, let us consider the contrapositive of the equivalence between Proposition \ref{prop: charc of ESE bip graph} $(i)$ and $(iii)$. This suggests that $G$ is not $\ESE$ if and only if there exists a vertex $u\in U$ such that for every non-empty subset $S\subseteq N(u)$ we have $|N(S)|\geq |S|$. 
The later condition together with Hall's Theorem implies the following characterization of graphs which are not $\ESE$. This formulation is indeed the most convenient one for our recognition algorithm, and thus worth mentioning separately. 

\begin{cor}\label{cor: not ESE bip graph}
A connected bipartite graph $G=(U \cup W, E)$ with $\vert U \vert < \vert W \vert$ is not $\ESE$ if and only if there exists $u\in U$ such that $N(u)$ is saturated by some (maximal) matching of $G$.
\end{cor}

%

\section{Recognition of ESE-graphs}\label{sec: Recog ESE}

The recognition of $\ESE$-graphs is trivially polynomial since checking equimatchability can be done in time $O(n^2m)$ for a graph with $n$ vertices and $m$ edges (see \cite{demange-ekim-equi}) and it is enough to repeat this check for every edge removal. This trivial procedure gives a recognition algorithm for $\ESE$-graphs in time $O(n^2m^2)$. However, using the characterization of $\ESE$-graphs, we can improve this time complexity in a significant way.

\begin{thm}\label{thm: ESE recognition}
$\ESE$-graphs can be recognized in time $O(\min(n^{3.376}, n^{1.5}m))$.
\end{thm}
\begin{proof}
Let us first note that  Theorems \ref{Thm: all 2 connected graphs} and \ref{thm: ESE with cut vertex is bip} imply that a connected $\ESE$-graph is either (2-connected) factor-critical or bipartite. \\
Remark \ref{rem: small FC ESE} exhibits all factor-critical $\ESE$-graphs with at most 5 vertices; it is clear that one can recognize if the given graph is isomorphic to one of them in constant time. Since factor-critical graphs have an odd number of vertices, there is no factor-critical $\ESE$-graph on 6 vertices. Besides, one can also check whether a given graph with at least 7 vertices is factor-critical $\ESE$ in linear time ($O(n+m)$) using the characterization given in Theorem \ref{thm: chracterization of edge-stable}. Indeed, to decide whether $G$ is isomorphic to $K_{2r+1}\setminus M$ for some matching $M$, it is enough to check if the minimum degree is at least $2r-1$. To decide whether $G$ admits an independent set $S$ of size at least 3 which is complete to $G\sm S$ and $\nu(G\sm S)=1$, one can simply search for a connected component of the complement of $G$ which is a clique (in linear time); if yes it is the unique candidate for the set $S$. Then to check whether $\nu(G\sm S)=1$, it is enough to notice that a graph has matching number 1 if and only if it is the disjoint union of a triangle and isolated vertices, or the disjoint union of a star and isolated vertices. To recognize these graphs, one can check whether the degree sequence of $G\sm S$ is one of $k,1, \ldots,1,0,\ldots ,0$ (where $k\geq 1$ and there are $k$ times 1) or $2,2,2,0,\ldots ,0$ where there is possibly no vertex of degree 0 at all. Clearly, these can be done in linear time.\\
Now, in order to decide whether a bipartite graph $G$ is $\ESE$, we use Corollary \ref{cor: not ESE bip graph}. For every vertex $u\in U$ where $U$ is the small part of the bipartition, compute a maximum matching of the bipartite graph $G[N(u)\cup N(N(u))]$; if $\nu(G[N(u)\cup N(N(u))])=|N(u)|$ for some $u\in U$, then it means that $G$ is not $\ESE$; otherwise it is $\ESE$. This check requires at most $n$ computations of a maximum matching in a bipartite graph, which can be done in time $O(n^{2.376})$  \cite{bip-matching-matrix} or in time $O(\sqrt{n}m)$ which runs in time $O(n^{2.5})$ in case of dense graphs but becomes near-linear for random graphs \cite{hopcroft-karp}. As this term dominates, it follows that the overall complexity of this recognition algorithm is $O(\min(n^{3.376}, n^{1.5}m))$.
\end{proof}

\section{Concluding Remarks}\label{sec: Conc remarks}

In this paper, we considered $\ESE$-graphs which are equimatchable such that the removal of any edge does not harm their equimatchability. We characterized $\ESE$-graphs under two exclusive categories: factor-critical and bipartite.

As shown in Proposition \ref{prop:unchanged matching number}, for an $\ESE$-graph $G$, the removal of an edge $e$ whose endpoints do not form a connected component of $G$ does not change the matching number of  $G$. This implies that for such an edge $e=uv$ and any maximal matching $M$, there is always a vertex $w$ exposed by $M$ which is adjacent to $u$ or $v$. 
In the line graph of $G$, such an edge corresponds to a vertex $x$  such that for every independent set $I$ of $G \sm (N(x)\cup \{x\})$, there exists some $y\in N(x)$ such that $I\cup \{y\}$ is independent. A vertex satisfying this property is known as a \emph{shedding vertex}. Shedding
vertices are strongly related to the combinatorial topology of independence complexes of graphs \cite{LM2016,wood2009}. In particular, they play a key role in identifying vertex decomposable graphs \cite{BC2014}. $\Shed(G)$ is defined as the set of shedding vertices of a graph $G$. 
In \cite{LM2016}, Levit and Mandrescu showed that a well-covered graph $G$ without isolated vertices has $\Shed(G)=V(G)$ if and only if $G$ is 1-well-covered. We observe that, every vertex of the line graph of an $\ESE$-graph with no component isomorphic to a $K_2$ is a shedding vertex.  It follows that, in the present paper, we characterized well-covered line graphs such that $\Shed(G)=V(G)$.  

After noticing that some well-covered graphs, such as $C_4$ and $C_7$, has no shedding vertex, finding all well-covered graphs having no shedding vertices was mentioned as an open problem in \cite{LM2016}. In terms of equimatchable graphs, this suggests to study the notion of criticality which is the opposite of stability.  Along this line, we introduce the notion of equimatchable  graphs such that the removal of any edge harms the equimatchability. More formally, for an equimatchable graph $G$, we say that $e \in E(G)$ is a \emph{critical-edge} if $G \sm e$ is not equimatchable. Note that if an equimatchable graph $G$ is not edge-stable, then it has a critical-edge. A graph $G$ is called \emph{edge-critical equimatchable}, denoted \emph{$\ECE$} for short, if $G$ is equimatchable and every $e \in E(G)$ is critical. We note that $\ECE$-graphs can be obtained from any equimatchable graph by recursively removing non-critical edges. We also remark that if $e$ is a non-critical edge whose endpoints do not form a connected component of $G$, then $\nu(G)=\nu(G-e)$ as already shown in Proposition \ref{prop:unchanged matching number}. 
By definition of $\ECE$-graphs, a graph $G$ with no $K_2$ component is $\ECE$ if and only if $L(G)$ is well-covered and $\Shed(L(G))=\emptyset$ where $L(G)$ is the line graph of $G$. Thus, the complete characterization of $\ECE$-graphs would enlighten the structure and the recognition of well-covered line graphs with no shedding vertex. We note that this question was explicitly mentioned as an open problem (for general well-covered graphs) in  \cite{LM2016}.

Our preliminary results show that $\ECE$-graphs are either 2-connected factor-critical or 2-connected bipartite or $K_{2t}$ for some $t\geq 2$. Moreover, 2-connected bipartite $\ECE$-graphs can be characterized using similar arguments as in the characterization of bipartite $\ESE$-graphs. On the other hand, 
 we know that smallest factor-critical $\ECE$-graphs have 7 vertices and there are exactly 4 such factor-critical $\ECE$-graphs, including the $C_7$. 
In order to complete the characterization of $\ECE$-graphs, and thus, answering the open problem in \cite{LM2016} for well-covered line graphs, it remains to complete the case of factor-critical $\ECE$-graphs, which we pose as an open question.

\begin{prob} \label{prob:ece}
Characterize/find all factor-critical $\ECE$-graphs.
\end{prob}

One can also extend the notion of stability and criticality of equimatchable graphs to vertex removals (note that in this case the link with 1-well-covered graphs is lost). An equimatchable graph $G$ is called \emph{vertex-stable} if $G-v$ is equimatchable for every $v\in V(G)$. In a way similar to $\ECE$-graphs, one can introduce \emph{vertex-critical equimatchable} graphs as equimatchable graphs which loose their equimatchability by the removal of any vertex. Our preliminary studies show that the description of vertex-stable equimatchable graphs and vertex-critical equimatchable graphs are much simpler compared to their edge counterparts, namely $\ESE$-graphs and $\ECE$-graphs respectively.

Last but not least, one can investigate some extensions of our work in terms of forbidden subgraphs or graph classes more broadly. It is well-known that line graphs are characterized by 9 forbidden subgraphs. Consequently, one may consider the effect of allowing some of these 9 forbidden subgraphs in terms of the structure and the recognition of 1-well-covered graphs. Along the same line, it could be interesting to consider the intersection of 1-well-covered graphs with some known minimal superclasses of line graphs such as quasi-line graphs and EPT-graphs (even though they are not defined by forbidden subgraphs).

\begin{prob} \label{prob:extension}
Find a characterization and/or an efficient recognition of 1-well-covered graphs belonging to the class of quasi-line graphs, or EPT-graphs, or graphs obtained by forbidding a proper subset of 9 forbidden subgraphs for line graphs. 
\end{prob}

\section*{Acknowledgments}
The support of 213M620 Turkish-Slovenian TUBITAK-ARSS Joint Research Project is greatly acknowledged. The authors are also grateful to the anonymous referees for their helpful suggestions to improve the paper.
\nocite{*}
\bibliographystyle{abbrv}

\end{document}